\documentclass[10pt]{article}
\usepackage[a4paper,left=24mm,right=24mm,top=23mm,bottom=24mm]{geometry}
\usepackage[T1]{fontenc}
\usepackage{amsmath,amssymb,amsfonts,mathtools,amsthm}
\usepackage{txfonts}
\usepackage{booktabs,array,enumitem,microtype,mathrsfs}
\usepackage[numbers,sort&compress]{natbib}
\usepackage{needspace}
\usepackage{url}
\usepackage[hidelinks,hypertexnames=false]{hyperref}
\hypersetup{
  bookmarksnumbered=true,
  pdfauthor={Yangcheng Li and Pingzhi Yuan},
  pdftitle={MDS codes containing a fixed subcode over finite field extensions},
  pdfsubject={Version 8}
}
\numberwithin{equation}{section}
\newtheorem{theorem}{Theorem}[section]
\newtheorem{proposition}[theorem]{Proposition}
\newtheorem{lemma}[theorem]{Lemma}
\newtheorem{corollary}[theorem]{Corollary}
\newtheorem{definition}[theorem]{Definition}
\newtheorem{remark}[theorem]{Remark}
\newtheorem{example}[theorem]{Example}

\newcommand{\F}{\mathbb F}
\newcommand{\PP}{\mathbb P}
\newcommand{\Gr}{\operatorname{Gr}}
\newcommand{\wt}{\operatorname{wt}}
\newcommand{\supp}{\operatorname{supp}}
\newcommand{\rank}{\operatorname{rank}}
\newcommand{\Span}{\operatorname{span}}
\newcommand{\sdef}{\operatorname{sdef}}
\newcommand{\cC}{\mathcal C}

\newcommand{\cP}{\mathcal P}

\begin{document}

\title{MDS codes containing a fixed subcode over finite field extensions}
\author{Yangcheng Li\thanks{Corresponding author. Email: \texttt{liyc@m.scnu.edu.cn}.}\quad
Pingzhi Yuan\thanks{Email: \texttt{yuanpz@scnu.edu.cn}.}\\[3pt]
\small School of Mathematical Sciences, South China Normal University\\
\small Guangzhou 510631, Guangdong, China}
\date{}
\maketitle

\begin{abstract}
Let $2\le k\le n$ and let $U\le\F_q^n$ be a fixed $(k-1)$-dimensional
linear code. We study the $[n,k]$ maximum distance separable (MDS) codes
over finite extensions of $\F_q$ that contain the scalar extension of $U$.
Our approach constructs an auxiliary linear code from \(U\). Its full-weight words characterize the desired MDS completions.
We prove that \(U\) admits an \([n,k]\) MDS completion over some finite extension of \(\mathbb F_q\) if and only if \(d(U)\ge n-k+1\). When this condition holds, the number of these MDS codes
over $\F_{q^m}$ is
\[
\#\left\{W\le\F_{q^m}^n:
U\otimes_{\F_q}\F_{q^m}\subseteq W,\ \dim_{\F_{q^m}} W=k,\ W\text{ is MDS}\right\}
=\frac{\chi_{\delta_GM(G)}(q^m)}{q^m-1},
\]
where $G$ is a generator matrix of $U$, $\delta_GM(G)$ is its represented
derived matroid, and $\chi_{\delta_GM(G)}$ is its characteristic polynomial.
In particular, the formula determines existence over each extension field
and the least extension degree for which such an MDS code exists.
Examples show that these counts may distinguish different representations
of the same column matroid, while monomially inequivalent codes may have
identical counts over every finite extension.
\end{abstract}



\section{Introduction}\label{sec:intro}

A linear $[n,k]_q$ code is maximum distance separable (MDS) if it
attains the Singleton bound, that is, if its minimum distance is
$n-k+1$. Equivalently, every $k$ columns of a generator matrix are
linearly independent. This criterion links MDS codes to arcs in finite
projective spaces and underlies much of their theory of existence,
construction, and classification; see
\cite{HuffmanPless,BallLavrauwArcs}. Those problems typically seek an
entire MDS code. A different problem arises when part of the desired
code has already been fixed: can one preserve that subcode while
completing it to an MDS code?

Our motivation comes from a covering problem involving cyclic
shifts of linear subspaces. Given a subspace
$V\leq\F_q^n$, one may ask whether the union of its
cyclic coordinate shifts is the whole space $\F_q^n$.
The invariant $h_q(n)$ is the largest codimension of
a subspace with this property. Li and Yuan proved
in \cite[Theorem~2.1]{LiYuanCyclicCovering} that
\[
h_q(n)=0
\quad\Longleftrightarrow\quad
\text{every nonzero cyclic code in $\F_q^n$
	contains a full-weight word}.
\]
This characterization has a natural interpretation in terms of MDS codes. Indeed, let \(C\) be a nonzero cyclic code and set \(U=C^{\perp}\). Since the dual of a cyclic code is again cyclic, \(U\) is also a cyclic code. Every nonzero \(a\in C\) defines a hyperplane \(a^{\perp}\) containing \(U\), and this hyperplane is an MDS \([n,n-1]\) code precisely when \(a\) has no zero coordinates. Consequently, \(h_q(n)=0\) is equivalent to the statement that every proper cyclic subspace of \(\mathbb{F}_q^n\) is contained in a hyperplane that is an MDS \([n,n-1]\) code. This motivates the study of MDS codes containing a prescribed subcode. We focus on the one-step case: given an arbitrary linear code \(U\), when can it be extended to an MDS code of dimension \(\dim U+1\), and how many such completions exist over finite field extensions?

The existence of nested MDS codes over sufficiently large scalar
extensions is connected with Pellikaan's construction of
error-correcting pairs \cite[Proposition~5.1]{PellikaanECP}. In the
parameter range considered below, applying this result to $U^\perp$
and dualizing provides a finite-extension MDS completion whenever the
necessary distance condition holds. Related existence results for
MDS seeds follow from \cite[Theorem~6.6]{CampsGorlaLopezCodeDistances}.
These results concern whether a completion can be found; our aim is to
count completions over each extension and determine the first
extension degree exactly.

Wu, Ding, and Chen \cite{WuDingChen} characterized when a one-coordinate extension of an MDS code remains MDS in terms of the covering radius and deep holes of its dual. Wang and Luo \cite{WangLuo2026} recently developed an iterative algorithm for extending non-Reed--Solomon MDS codes by adding a coordinate. Our work studies dimension-increasing MDS completions at fixed length and determines their exact numbers over finite field extensions.

Fix integers $2\le k\le n$ and a prescribed $(k-1)$-dimensional
subcode \(U\le\F_q^n.\) Throughout this paper, we refer to this fixed subcode \(U\) as the \emph{seed}, since it serves as the starting point from which the larger MDS code is constructed.
For $Q=q^m$, let
\[
U_Q=U\otimes_{\F_q}\F_Q
\cong \operatorname{span}_{\F_Q}(U)
\subseteq \F_Q^n
\]
be the scalar extension of $U$ to $\F_Q$.
Here the natural identification is given by
$u\otimes a\mapsto au$.
Thus tensoring with $\F_Q$ amounts to allowing
coefficients from $\F_Q$ in linear combinations
of vectors of $U$.
In particular,
$\dim_{\F_Q}U_Q=\dim_{\F_q}U=k-1$.
We count
\begin{equation}\label{eq:definition-N}
N_U(Q)=
\#\bigl\{W\le\F_Q^n:U_Q\subseteq W,\ \dim_{\F_Q}W=k,
\ W\text{ is MDS}\bigr\}.
\end{equation}
A code $W$ in this set is a \emph{one-step MDS completion} of $U_Q$:
``one-step'' means that the dimension is increased by one while the
length is unchanged. We work in fixed coordinate positions: \eqref{eq:definition-N} counts
actual subspaces, not monomial-equivalence classes.

The problem naturally leads to three questions:
\begin{enumerate}
	\item[(i)] 
	Does there exist a finite extension $\F_{q^m}$
	such that $N_U(q^m)>0$?
	
	\item[(ii)] 
	Can we determine $N_U(q^m)$ explicitly
	for every $m\ge1$?
	
	\item[(iii)] 
	If such an extension exists, what is the smallest
	$m\ge1$ for which $N_U(q^m)>0$?
\end{enumerate}

Our approach is inspired by the full-weight perspective arising from the cyclic covering problem. With a generator matrix of the seed fixed, the
maximal minors of an appended generator matrix depend linearly on the
appended row. Their values form the Pl\"ucker extension code
$\cP(U)$ of Definition~\ref{def:plucker-extension-code}.
Over each scalar extension, its full-weight words, modulo nonzero
scalars, correspond exactly to the desired MDS completions, as established in
Theorem~\ref{thm:full-MDS}.

This correspondence leads to the following results. First, a finite-extension MDS completion exists if and only if \(d(U) \geq r = n - k + 1\); equivalently, the seed \(U\) is MDS or almost MDS. Second, when this condition holds, for every \(Q = q^m\),
\[
N_U(Q)
=\frac{\#\{c\in\cP(U)_Q:\wt(c)=\binom nk\}}{Q-1}
=\frac{\chi_{\delta_G M(G)}(Q)}{Q-1},
\]
where \(\delta_G M(G)\) is the represented derived matroid of a generator matrix \(G\) of \(U\). Consequently, existence over \(\mathbb{F}_{q^m}\) is equivalent to \(\chi_{\delta_G M(G)}(q^m) > 0\). Third, the admissible extension degrees form a threshold interval \(m \geq m_0(U)\), and \(m_0(U)\) has a support-cover interpretation and the bound \(m_0(U) \leq \min\{r,\max\{1, \lceil \log_q \binom{n}{k} \rceil\}\}\). 

The paper is organized as follows. Section~\ref{sec:prelim} collects
the coding-theoretic, matroidal, arrangement-theoretic, and metric
background used later. Section~\ref{sec:main} develops the full-weight
correspondence, identifies the derived matroid, and derives the counting
and extension-degree results together with their covering-radius
interpretation. Section~\ref{sec:benchmark} presents two explicit
examples illustrating the information captured by the counts and their
limitations. Section~\ref{sec:conclusion} summarizes the main results.

\section{Preliminaries}\label{sec:prelim}

This section fixes notation and recalls the tools used in the proof.

\subsection{Linear codes, MDS codes, and scalar extensions}
\label{subsec:prelim-codes}

Throughout, $[n]=\{1,\ldots,n\}$ and $V=\F_q^n$. We use row
vectors for codewords and circuit vectors. For a matrix $A$, we use
the row-vector convention $\ker A=\{x:Ax^{\mathsf T}=0\}$, identifying
the usual column kernel with its transpose. Matrix columns and
syndrome vectors are viewed as column vectors. For a vector $x$, its
Hamming weight $\wt(x)$ is the number of its nonzero coordinates and
its support is denoted by $\supp(x)$. A word has \emph{full weight}, if all its coordinates are nonzero,
that is, if its weight equals its length. The minimum distance of a nonzero linear code $C$ is
$d(C)=\min_{0\ne c\in C}\wt(c)$. An $[n,s]_q$ code satisfies the
Singleton bound $d(C)\le n-s+1$. Its Singleton defect is
\[
\sdef(C)=n-s+1-d(C).
\]
A code is MDS if this defect is zero and almost MDS (AMDS) if the
defect is one. A $k\times n$ generator matrix generates an MDS code
if and only if each choice of $k$ columns is linearly independent
\cite{HuffmanPless}.

Two linear codes of the same length over the same field are
\emph{monomially equivalent} if one is obtained from the other by
permuting coordinates and multiplying each coordinate by a nonzero
field element.

For an extension field $K/\F_q$, write
$C_K=C\otimes_{\F_q}K\le K^n$. For a finite extension of order
$Q=q^m$, we abbreviate $C_{\F_Q}$ to $C_Q$; in particular,
$U_{q^m}=U_{\F_{q^m}}$. Submatrices of a matrix over $\F_q$
have the same ranks after scalar extension, a fact that immediately
controls the minimum distance.

\begin{lemma}\label{lem:zero-rank}
If $C\le\F_q^n$ has dimension $s\ge1$ and generator matrix $A$, then
\begin{equation}\label{eq:distance-rank}
 d(C)=n-\max\{|S|:\rank A_S<s\}.
\end{equation}
Here, for any subset \(S \subseteq [n]\), \(A_S\) denotes the submatrix of \(A\) consisting of the columns indexed by \(S\).
In particular, Hamming minimum distance is unchanged by scalar extension.
\end{lemma}

\begin{proof}
	A word \(xA\) (with \(x \in \mathbb{F}_q^s\)) vanishes on a coordinate set \(S\) exactly when \(xA_S=0\). Such a nonzero coefficient vector \(x\) exists exactly when the rows of \(A_S\) are linearly dependent, i.e., \(\rank A_S < s\). 
	
	The Hamming weight of \(xA\) is precisely \(n - |S|\), where \(S\) is its zero set. Thus minimizing the weight over all nonzero \(xA\) is equivalent to maximizing the size of a set \(S\) such that \(\rank A_S < s\). This proves \eqref{eq:distance-rank}.
	
	Finally, since extending the field does not change the rank of any submatrix, the minimum distance remains unchanged under scalar extension.
\end{proof}

\subsection{One-step directions and Pl\"ucker coordinates}
\label{subsec:prelim-plucker}

Fix a \((k-1)\)-dimensional code \(U \le V = \mathbb{F}_q^n\), where \(2\le k\le n\), and put
\[
r=n-k+1,\qquad E_U=V/U.
\]
Let \(K/\mathbb{F}_q\) be a finite extension and set \(V_K = K^n\). 
A \(k\)-dimensional code containing \(U_K\) is necessarily of the form \(U_K+\langle v\rangle\) for some \(v\notin U_K\). Two such vectors produce the same code if and only if their nonzero classes in \(K^n/U_K\) differ by a scalar in \(K^\times\). Thus, one-step directions are the points of \(\PP(K^n/U_K)\), a projective space of dimension \(r-1\). In other words, a one-step direction is an equivalence class of candidate vectors \(v\), where we identify vectors that yield the same code.

On the other hand, the completed code itself, \(W = U_K + \langle v \rangle\), is a \(k\)-dimensional subspace of \(K^n\). Recall that the Grassmannian \(\Gr(k,n)\) parametrizes the \(k\)-dimensional subspaces of a fixed \(n\)-dimensional vector space. A full-rank \(k\times n\) matrix specifies a point of \(\Gr(k,n)\) by its row space. 

More explicitly, there are exactly \(\binom{n}{k}\) maximal minors (each corresponding to a choice of \(k\) columns). Ordered lexicographically, these minors form a vector in \(K^{\binom{n}{k}}\). If \(A\) and \(A'\) are two full-rank \(k\times n\) matrices with the same row space, then \(A'=SA\) for some \(S\in\operatorname{GL}_k(K)\). For every \(k\)-subset \(I\),
\[
\det A'_I=\det(S)\det A_I.
\]
Thus all maximal minors are multiplied by the same nonzero scalar, so they determine a well-defined point in \(\mathbb{P}^{\binom{n}{k}-1}\). These projective coordinates are the Plücker coordinates of the subspace \cite{HarrisAG}.

In this geometric language, Proposition \ref{prop:MDS-plucker} states that a subspace corresponds to an MDS code if and only if all of its Plücker coordinates are nonzero.

\begin{proposition}\label{prop:MDS-plucker}
A full-rank $k\times n$ matrix generates an MDS code if and only
if each of its $k\times k$ minors is nonzero. Equivalently, every
Pl\"ucker coordinate of its row space in $\Gr(k,n)$ is nonzero.
\end{proposition}

\begin{proof}
A code is MDS if and only if every $k$ columns of a generator matrix
are linearly independent. Indeed, a dependence among some $k$ columns
is equivalent to the existence of a nonzero codeword vanishing on those
$k$ coordinates. The final statement follows because the maximal minors
are precisely the Pl\"ucker coordinates
\cite{HuffmanPless,HarrisAG}.
\end{proof}

For the particular row-extension problem, only the elementary
minor criterion above is necessary.

\begin{example}\label{ex:three-coordinates}
Let $U=\langle(1,1,1)\rangle\le\F_2^3$, and extend scalars to
$K=\F_{2^m}$ of order $Q$. Appending $v=(v_1,v_2,v_3)$ gives
\[
\begin{pmatrix}1&1&1\\v_1&v_2&v_3\end{pmatrix}.
\]
Its three maximal minors are $v_2-v_1$, $v_3-v_1$, and $v_3-v_2$.
Thus the resulting $[3,2]_Q$ code is MDS precisely when the three
entries of $v$ are pairwise distinct. There are
$Q(Q-1)(Q-2)$ such vectors. Since \(U_K\) has size \(Q\), adding any element of \(U_K\) to \(v\) yields \(Q\) distinct vectors in the same coset. Additionally, scaling this nonzero coset by any element in \(K^\times\) gives \(Q - 1\) distinct cosets that define the same completion. Hence
\[
N_U(Q)=\frac{Q(Q-1)(Q-2)}{Q(Q-1)}=Q-2.
\]
In particular, there is no completion over $\F_2$, but there are
two over $\F_4$. 
\end{example}

\subsection{Matroids, circuits, and represented derived matroids}
\label{subsec:prelim-matroids}

Let \(A=(a_e)_{e\in E}\) be a matrix over \(\mathbb{F}_q\). Its column matroid \(M(A)\) has ground set \(E\): a subset \(S\subseteq E\) is independent exactly when the columns indexed by \(S\) are linearly independent. The rank function is \(r_{M(A)}(S)=\operatorname{rank} A_S\), and the rank of the matroid itself is $r(M(A)) = r_{M(A)}(E) = \operatorname{rank} A$. A \emph{circuit} is a minimal dependent set. In particular, a zero column gives a loop (a one-element circuit), and two nonzero proportional columns form a two-element circuit. The uniform matroid \(U_{s,n}\) is the rank-\(s\) matroid in which every subset of size at most \(s\) is independent. A \emph{flat} is a subset \(F \subseteq E\) such that if an element \(e \in E\) has its column in the span of the columns indexed by \(F\), then \(e \in F\) (equivalently, \(F\) is closed under the matroid closure). We use these notions in their usual representable-matroid sense \cite{OxleyBook}.

For a matrix \(G\) and a circuit \(C\) of \(M(G)\), minimal dependence gives a vector \(c_C\in\ker G\), unique up to a nonzero scalar, with \(\operatorname{supp}(c_C)=C\). Such a vector is called a \emph{circuit vector}. If \(M(G)\) has \(t\) circuits, the \emph{represented derived matroid} \(\delta_GM(G)\) is the column matroid of the \(n \times t\) matrix whose columns are the transposes \(c_C^{\mathsf T}\) of chosen circuit vectors, one for each circuit \(C\) \cite{OxleyWang}. Its ground set is therefore the set of circuits of \(M(G)\), not the original set of coordinate positions. Different non-equivalent representations of an abstract matroid can yield different represented derived matroids; for this reason we retain the subscript \(G\).

For the row code \(U\) of \(G\), note the canonical identification
\[
(V/U)^*\simeq U^\perp=\ker G,
\]
where a vector \(c\in U^\perp\) acts on \(v+U\) by the dot product \(c\cdot v\). This elementary duality allows linear functionals on the quotient space to be evaluated via dot products in the original space. It is precisely the link between circuit vectors and the normals to the obstruction hyperplanes constructed later.

\begin{example}
	Consider a \(2 \times 4\) matrix over \(\mathbb{F}_2\):
	\[
	G = \begin{pmatrix} 1 & 0 & 1 & 1 \\ 0 & 1 & 1 & 0 \end{pmatrix}. 
	\]
	Let \(E=\{1,2,3,4\}\) be the ground set. Since \(G\) has full row rank, its column matroid \(M(G)\) has rank \(2\).
	\begin{itemize}
		\item \textbf{Circuits and flat:} The columns \(1\) and \(4\) are identical, so \(C_1=\{1,4\}\) is a two-element circuit. Similarly, \(C_2=\{1,2,3\}\) and \(C_3=\{2,3,4\}\) are three-element circuits. The set \(\{1,4\}\) is a flat, as it contains all columns in its span. (If we remove column 4, the remaining matroid becomes the uniform matroid \(U_{2,3}\)).
		\item \textbf{Derived matroid:} The corresponding circuit vectors (in \(\ker G\)) are \(v_1=(1,0,0,1)\), \(v_2=(1,1,1,0)\), and \(v_3=(0,1,1,1)\). The represented derived matroid \(\delta_G M(G)\) is the column matroid of the \(4 \times 3\) matrix \([v_1^{\mathsf T} \mid v_2^{\mathsf T} \mid v_3^{\mathsf T}]\). Since \(v_1+v_2+v_3=0\), the set \(\{C_1,C_2,C_3\}\) forms a circuit in the derived matroid itself.
		\item \textbf{Duality:} Take \(c=v_2=(1,1,1,0) \in \ker G\) and \(x=(1,1,0,0)\). For any representative \(x' = x + u\) of the coset \(x+U\) (e.g., \(u=(1,0,1,1) \in U\)), the dot product \(c \cdot x'\) remains \(0\). This illustrates the well-defined action of \((V/U)^* \simeq \ker G\).
	\end{itemize}
\end{example}

\subsection{Hyperplane arrangements and characteristic polynomials}
\label{subsec:prelim-arrangements}

A \emph{central hyperplane arrangement} in a vector space \(V\) is a finite family of hyperplanes (i.e., kernels of nonzero linear forms \(V \to \mathbb{F}_q\)). Proportional forms define the same hyperplane; an arrangement will be understood as a collection of distinct hyperplanes. Its \emph{normal matroid} is represented by choosing one nonzero normal vector for each hyperplane. The arrangement has full rank if these normals span the dual space \(V^*\).

For a rank-\(s\) matroid \(M\) on a finite ground set \(E\), its \emph{characteristic polynomial} is defined by
\begin{equation}\label{eq:chi-subset}
	\chi_M(t)=\sum_{A\subseteq E}(-1)^{|A|}t^{s-r_M(A)}.
\end{equation}
If \(M\) has a loop, this polynomial is identically zero. If \(M\) is loopless and represented by a full-row-rank matrix over \(\F_q\) with \(s\) rows, then \(\chi_M(q^m)\) counts the vectors in \(\F_{q^m}^s\) that avoid the associated hyperplanes, as proved in Proposition~\ref{prop:full-support-count}. 

\begin{example}
	Let \(M\) be the uniform matroid \(U_{2,3}\). Here the ground set is \(E = \{1, 2, 3\}\) and the rank is \(s = 2\). By definition, every subset of size at most \(2\) is independent, while the only dependent set is the full set \(E\) itself. 
	
	To compute the characteristic polynomial \(\chi_M(t)\) using \eqref{eq:chi-subset}, we group the \(2^3 = 8\) subsets \(A \subseteq E\) by their sizes:
	\begin{itemize}
		\item \(|A| = 0\): \(A = \emptyset\), \(r_M(A) = 0\). The term is \((-1)^0 t^{2-0} = t^2\).
		\item \(|A| = 1\): There are \(3\) such subsets, and all have rank \(1\). The terms are \(3 \cdot (-1)^1 t^{2-1} = -3t\).
		\item \(|A| = 2\): There are \(3\) such subsets, and all have rank \(2\). The terms are \(3 \cdot (-1)^2 t^{2-2} = 3\).
		\item \(|A| = 3\): There is \(1\) such subset (the whole set), which is dependent, so its rank is \(2\). The term is \((-1)^3 t^{2-2} = -1\).
	\end{itemize}
	
	Summing these terms gives the characteristic polynomial:
	\[
	\chi_M(t) = t^2 - 3t + 3 - 1 = t^2 - 3t + 2 = (t-1)(t-2).
	\]
\end{example}

\subsection{Coset distances, covering radii, and syndromes}
\label{subsec:prelim-cosets}

Let \(K/\mathbb{F}_q\) be an extension field. For a coset \(\bar{v} = v + U_K\) in \(K^n/U_K\), define its \emph{coset distance} and the \emph{covering radius} by
\begin{equation}\label{eq:coset-distance}
	\delta_{U,K}(\bar{v}) = \min_{u\in U_K}\operatorname{wt}(v-u), \qquad
	\rho(U_K) = \max_{\bar{v}\in K^n/U_K}\delta_{U,K}(\bar{v}).
\end{equation}
A coset at distance \(\rho(U_K)\) is called a \emph{deep-hole coset}; its elements are called deep holes. 

An \emph{information set} of \(U_K\) is a subset \(S\subseteq[n]\) of size \(k-1\) for which the coordinate projection \(U_K\to K^S\) is an isomorphism; equivalently, the columns indexed by \(S\) form a nonsingular square submatrix of a generator matrix of \(U_K\).

To bound this distance, recall that \(U_K\) has dimension \(k-1\), so an information set of \(U_K\) has size \(k-1\). For any candidate vector \(v\), we can subtract a suitable word of \(U_K\) to make it zero on such an information set. After this matching, only the remaining \(n - (k-1) = n-k+1 = r\) coordinates can possibly be nonzero. Consequently, every coset has a representative of weight at most \(r\), yielding the redundancy bound:
\begin{equation}\label{eq:redundancy}
	\delta_{U,K}(\bar{v}) \le \rho(U_K) \le r.
\end{equation}

To explicitly represent these cosets, let \(H\) be an \(r \times n\) parity-check matrix of \(U\), with full row rank, so that \(U=\{v\in\F_q^n:Hv^{\mathsf T}=0\}\). The syndrome of \(v\) is defined as \(Hv^{\mathsf T}\). Since the kernel of this map is precisely \(U_K\), the first isomorphism theorem provides a canonical identification between the abstract quotient space and the concrete vector space \(K^r\):
\[
K^n/U_K \xrightarrow{\sim} K^r,
\qquad v+U_K \longmapsto Hv^{\mathsf T}.
\]
This identification represents quotient vectors in \(K^n/U_K\) as syndrome vectors in \(K^r\); the corresponding one-step directions are points of \(\PP(K^r)\). In this syndrome space, the obstruction hyperplanes can be described geometrically using the columns of \(H\) (as detailed in Section \ref{sec:covering}).

\section{Main results}\label{sec:main}

\subsection{Statement of the main theorem}\label{subsec:main-statement}

We first state the main theorem. Its proof will be completed in Section \ref{subsec:main-counting}.

\begin{theorem}\label{thm:critical-MDS}
Let $U\le\F_q^n$ have dimension $k-1$, $2\le k\le n$, and let \(G\) be a \((k-1) \times n\) generator matrix of \(U\). Put $r=n-k+1$.

If $d(U)\le r-1$, then no finite scalar extension of $U$ has an $[n,k]$ one-step MDS completion.

Assume $d(U)\ge r$.  Then for every $m\ge1$, with $Q=q^m$,
\begin{equation}\label{eq:main-count}
 \#\{W\le\F_Q^n:U_Q\subset W,\ \dim W=k,\ W\text{ MDS}\}
 =\frac{\chi_{\delta_GM(G)}(Q)}{Q-1}.
\end{equation}
Equivalently, $U_{q^m}$ admits such a completion if and only if
$\chi_{\delta_GM(G)}(q^m)>0$. 
\end{theorem}

\subsection{The Pl\"ucker extension code and the full-support criterion}
\label{subsec:main-plucker}

We first work over the base field. Choose a full-rank \((k-1)\times n\) generator matrix \(G\) of \(U\). Appending a row \(v\in V=\F_q^n\) gives
\begin{equation}\label{eq:appended-matrix}
	G(v)=\begin{pmatrix}G\\v\end{pmatrix}.
\end{equation}
For \(v\notin U\), the row space of \(G(v)\) is a candidate completion. For a \(k\)-subset \(I=\{i_1<\cdots<i_k\}\subseteq[n]\), let \(G(v)_I\) be the submatrix with column set \(I\), and define
\begin{equation}\label{eq:ell}
	\ell_I:E_U=V/U\longrightarrow\F_q,
	\qquad \ell_I(v+U)=\det G(v)_I.
\end{equation}

To understand this form, expanding the determinant along the last row shows that \(\ell_I\) is a linear functional in \(v\). Furthermore, it is well-defined on the quotient space \(V/U\): replacing \(v\) by \(v+u\) (for some \(u \in U\)) amounts to adding a row from \(U\) to the last row, which leaves the determinant unchanged. If \(G\) is replaced by another generator matrix of the same \(U\), every \(\ell_I\) is multiplied by the same nonzero scalar (the determinant of the change-of-basis matrix). Thus the kernels of these forms are intrinsic to the code \(U\). A nonzero form defines a proper hyperplane, whereas a zero form has kernel equal to all of \(E_U\).

By Proposition~\ref{prop:MDS-plucker}, an appended generator
matrix gives an MDS code exactly when all its maximal minors are
nonzero. Collecting these minors as coordinates of an auxiliary linear
code turns the completion problem into a full-weight problem. The code
has length \(\binom nk\), with one coordinate for each maximal minor.

\begin{definition}\label{def:plucker-extension-code}
	Define
	\begin{equation}\label{eq:plucker-code}
		\Phi_U:E_U\longrightarrow \F_q^{\binom{n}{k}},
		\qquad
		\bar v\longmapsto
		\bigl(\ell_I(\bar v)\bigr)_{|I|=k},
	\end{equation}
	and let
	\[
	\mathcal P(U):=\operatorname{im}\Phi_U.
	\]
	We call \(\mathcal P(U)\) the \emph{Pl\"ucker extension code} of \(U\).
\end{definition}

Every word of \(\mathcal P(U)\) records the maximal minors of \(G(v)\) for one quotient vector. Moreover, the map \(\Phi_U\) is injective. Indeed, if \(\Phi_U(\bar v)=0\), all \(k\times k\) minors of \(G(v)\) vanish, so \(\operatorname{rank} G(v)<k\). Since \(G\) has rank \(k-1\), this forces \(v\in U\) and thus \(\bar v=0\). It follows that
\[
\dim\mathcal P(U)=\dim(V/U)=r.
\]
For an extension field \(K/\F_q\), the natural identification
\(E_U\otimes_{\F_q}K\simeq K^n/U_K\) gives the scalar-extended forms
\[
\ell_{I,K}:K^n/U_K\longrightarrow K,
\qquad \ell_{I,K}(v+U_K)=\det G(v)_I.
\]
These forms are well-defined because adding any \(u\in U_K\) to the last row leaves the determinant unchanged. When the field is clear, we also write \(\ell_I\) for \(\ell_{I,K}\). Write \(\mathcal P(U)_K=\mathcal P(U)\otimes_{\F_q}K\) and denote by \(\Phi_{U,K}\) the scalar extension of \(\Phi_U\). The corresponding maximal minors are the coordinates of \(\Phi_{U,K}(\bar v)\).

\begin{example}\label{ex:plucker-construction}
	Consider the code
	\[
	U=\langle(1,1,1)\rangle\le\F_2^3
	\]
	from Example~\ref{ex:three-coordinates}.
	A generator matrix of $U$ is $G=(1,1,1)$.
	Every coset in $\F_2^3/U$ has a unique representative
	$v=(0,a,b)$, where $a,b\in\F_2$.
	Appending this row gives
	\[
	G(v)=
	\begin{pmatrix}
		1&1&1\\
		0&a&b
	\end{pmatrix}.
	\]
	Its three maximal minors, indexed by
	$\{1,2\}$, $\{1,3\}$, and $\{2,3\}$,
	are $a$, $b$, and $a+b$, respectively.
	Thus
	\[
	\Phi_U(v+U)=(a,b,a+b),
	\]
	and hence
	\[
	\mathcal P(U)
	=\{(a,b,a+b):a,b\in\F_2\}.
	\]
	In particular, $\mathcal P(U)$ is a $[3,2]_2$
	linear code with generator matrix
	\[
	\begin{pmatrix}
		1&0&1\\
		0&1&1
	\end{pmatrix}.
	\]
	It has no full-support word over $\F_2$.
	However, over $\F_4$, the scalar-extended
	code contains the full-weight word
	$(1,\alpha,1+\alpha)$, where
	$\alpha\in\F_4\setminus\F_2$.
\end{example}

The following correspondence is the central reduction used throughout
the paper.

\begin{theorem}\label{thm:full-MDS}
	For every extension field \(K/\F_q\), the following are equivalent:
	\begin{enumerate}[label=\textup{(\roman*)},leftmargin=2.5em]
		\item \(\mathcal{P}(U)_K\) contains a word of weight \(\binom{n}{k}\);
		\item \(U_K\) is contained in an \([n,k]_K\) MDS code.
	\end{enumerate}
	More precisely, the MDS completions of \(U_K\) are in bijection with the one-dimensional \(K\)-subspaces of \(\mathcal{P}(U)_K\) spanned by full-weight words.
\end{theorem}

\begin{proof}
	After scalar extension, every word of \(\mathcal{P}(U)_K\) has the form
	\[
	\Phi_{U,K}(\bar v)
	=\bigl(\ell_I(\bar v)\bigr)_{|I|=k},
	\]
	for some \(\bar v\in E_U\otimes_{\F_q}K\). This word has weight
	\(\binom{n}{k}\) if and only if all the corresponding \(k\times k\)
	maximal minors of \(G(v)\) are nonzero. By
	Proposition~\ref{prop:MDS-plucker}, this is equivalent to the row space
	\(U_K+\langle v\rangle\) being an \([n,k]_K\) MDS code.
	
	To see the bijection, note that two nonzero vectors \(\bar v\) and \(\bar v'\) determine the
	same completion if and only if they differ by a nonzero scalar in
	\(K^\times\). Since \(\Phi_{U,K}\) is injective, the same holds for their
	images in \(\mathcal{P}(U)_K\), which completes the proof.
\end{proof}

\begin{remark}
	Although not required for the counting proof, the same construction admits a useful basis-free description via exterior algebra. Let \(u_1,\ldots,u_{k-1}\) be the rows of \(G\) and set
	\[
	\xi_U=u_1\wedge\cdots\wedge u_{k-1}
	\in\bigwedge^{k-1}U.
	\]
	The map
	\[
	V/U\longrightarrow\bigwedge^kV,\qquad
	v+U\longmapsto\xi_U\wedge v
	\]
	is well-defined and injective, and its coordinates with respect to the standard exterior basis \(\{e_I\}_{|I|=k}\) are precisely \(\det G(v)_I\). Consequently, its image is the canonically embedded \(\mathcal{P}(U)\), whose projectivization is the linear Schubert subspace
	\[
	\{[W]\in\operatorname{Gr}(k,V):U\subseteq W\}.
	\]
	See \cite{HarrisAG} for this Grassmannian description.
	When \(k\le n-1\), the canonical embedding even determines the seed:
	\[
	U=\{x\in V:x\wedge c=0\text{ for all }c\in\mathcal{P}(U)\}.
	\]
	Indeed, every \(x\in U\) annihilates \(\xi_U\wedge v\). Conversely, if \(x\notin U\), choose \(v\notin U+\langle x\rangle\); then \(x\wedge\xi_U\wedge v\ne0\). This observation provides additional geometric insight, though it is not required for the subsequent arguments.
\end{remark}

\subsection{The distance obstruction}\label{subsec:main-obstruction}

Before counting the allowed directions, we must ensure that none of the linear forms \(\ell_I\) vanish identically. Indeed, by Theorem~\ref{thm:full-MDS}, an MDS completion requires a full-weight word in the Plücker extension code, meaning all coordinates must be nonzero. A zero form would impose an impossible condition over {every} scalar extension. For any subset \(I\subseteq[n]\), we denote by \(G_I\) the submatrix of \(G\) formed by the columns indexed by \(I\).

\begin{lemma}\label{lem:determinantal-form}
	For every subset \(I\subseteq[n]\) of size \(k\),
	\begin{equation}\label{eq:loop-rank}
		\ell_I=0
		\quad\Longleftrightarrow\quad
		\rank G_I\le k-2.
	\end{equation}
\end{lemma}

\begin{proof}
	Write \(I=\{i_1<\cdots<i_k\}\). Expanding along the last row gives
	\[
	\det\big( G(v)_I \big)
	=
	\sum_{t=1}^k
	(-1)^{k+t}
	\big(\det G_{I\setminus\{i_t\}}\big)\,v_{i_t}.
	\]
	Thus, the form \(\ell_I\) vanishes identically if and only if all the \((k-1)\times(k-1)\) minors of \(G_I\) vanish. This occurs precisely when \(\rank G_I\le k-2\).
\end{proof}

Combining Lemma~\ref{lem:determinantal-form} with the distance-rank formula \eqref{eq:distance-rank}, we obtain a purely structural criterion for the seed code \(U\). For an integer \(k \ge 0\), we denote by \(\binom{[n]}{k}\) the family of all \(k\)-element subsets of \([n]\).

\begin{corollary}\label{cor:structural-obstruction}
	The following are equivalent:
	\begin{enumerate}[label=\textup{(\roman*)},leftmargin=2.5em]
		\item \(\ell_I\ne0\) for every \(I\in\binom{[n]}{k}\);
		\item \(\rank G_I=k-1\) for every \(I\in\binom{[n]}{k}\);
		\item \(d(U)\ge r\).
	\end{enumerate}
	If these conditions fail, no scalar extension of \(U\) admits an \([n,k]\) MDS completion.
\end{corollary}

\begin{proof}
	First, Lemma~\ref{lem:determinantal-form} directly establishes the equivalence (i)\(\Leftrightarrow\)(ii). 
	
	Next, by \eqref{eq:distance-rank}, condition (ii) means that the zero set of any nonzero word in \(U\) has size at most \(k-1\). This is precisely the statement that \(d(U)\ge r\), proving (ii)\(\Leftrightarrow\)(iii).
	
	Finally, if any determinant form is identically zero, it remains zero after any scalar extension. By Theorem~\ref{thm:full-MDS}, this obstructs MDS completion over all extension fields.
\end{proof}

Since \(U\) has dimension \(k-1\), the Singleton bound yields \(d(U)\le n-(k-1)+1 = r+1\). Therefore, the condition \(d(U)\ge r\) means exactly that \(U\) is either MDS or almost MDS (AMDS). The condition is exactly that no coordinate of \(\cP(U)\) vanishes identically; a full-weight word may nevertheless require scalar extension.

\subsection{Identification of the derived matroid}
\label{subsec:main-identification}\label{sec:derived}

The full-weight condition requires all coordinate forms \(\ell_I\)
to be nonzero on the same quotient vector. To count such vectors, we
identify the distinct nonzero forms up to scalar and determine the
linear dependencies among them. The following construction relates
these forms to the circuits of the seed's column matroid.

While \(\ell_I\) is a linear form on the quotient space \(V/U\), the duality \((V/U)^*\simeq\ker G\) from Section~\ref{subsec:prelim-matroids} allows us to represent these forms concretely. For \(I=\{i_1<\cdots<i_k\}\subseteq[n]\), we define a vector \(c_I\in\F_q^n\) by
\[
(c_I)_{i_t}=(-1)^{k+t}\det G_{I\setminus\{i_t\}}
\quad(1\le t\le k),
\qquad (c_I)_j=0\quad(j\notin I).
\]
By Laplace expansion, \(c_I\) is a cofactor vector representing the form \(\ell_I\):
\begin{equation}\label{eq:cofactor}
	Gc_I^{\mathsf T}=0,\qquad
	\det G(v)_I=c_I\cdot v.
\end{equation}
This shows that \(c_I\in\ker G=U^\perp\) is a concrete representative of the form \(\ell_I\). When \(c_I\ne0\), it is a normal vector to the hyperplane \(\ker\ell_I\). Some \(c_I\) may be zero, and different nonzero \(c_I\) may be proportional, in which case they define the same hyperplane. The following theorem precisely identifies these distinct nonzero normal lines using the column circuits of \(G\).

\begin{theorem}\label{thm:derived-identification}
	Let \(G\) have full row rank \(k-1\), and let \(U\) be its row space. For every \(k\)-subset \(I\subseteq[n]\):
	\begin{enumerate}[label=\textup{(\roman*)},leftmargin=2.5em]
		\item \(c_I=0\) if and only if \(\rank G_I\le k-2\);
		\item if \(c_I\ne0\), then \(\supp(c_I)\) is the unique circuit of \(M(G)\) contained in \(I\);
		\item every circuit vector of \(M(G)\) occurs, up to a nonzero scalar, among the nonzero \(c_I\).
	\end{enumerate}
	
	Consequently, after deleting the zero forms and identifying proportional forms, the normal matroid of the arrangement defined by the nonzero \(\ell_I\) is exactly \(\delta_GM(G)\). This arrangement has rank \(r=n-k+1\). 
	
	If \(d(U)\ge r\), all the forms \(\ell_I\) are nonzero. Over any extension field, the projectivization of the complement of the scalar-extended arrangement parametrizes the MDS completions. Conversely, if \(d(U)<r\), some \(\ell_I\) vanishes identically, meaning no MDS completion exists over any scalar extension.
\end{theorem}

\begin{proof}
	{Part (i)} is a direct consequence of Lemma~\ref{lem:determinantal-form}.
	
	{For Part (ii)}, assume \(c_I\neq0\). Then \(\rank G_I=k-1\), which implies that the kernel \(\ker G_I\) is one-dimensional. The restriction of \(c_I\) to \(I\) is a nonzero vector in this kernel, meaning the columns of \(G_I\) indexed by \(\supp(c_I)\) are linearly dependent. 
	
	To see that this dependence is minimal, suppose a proper subset of \(\supp(c_I)\) were also dependent. Then there would exist a nonzero vector in \(\ker G_I\) whose support is strictly smaller than \(\supp(c_I)\). But since \(\ker G_I\) is one-dimensional, every nonzero vector in it is merely a scalar multiple of \(c_I\), and scaling cannot change the support size. This is a contradiction. Hence, \(\supp(c_I)\) is a minimal dependent set (a circuit). Moreover, the one-dimensionality of \(\ker G_I\) guarantees that this is the unique circuit contained in \(I\).
	
	{For Part (iii)}, let \(C\) be an arbitrary circuit of \(M(G)\), and choose an element \(e\in C\). Since \(C\) is minimal, \(C\setminus\{e\}\) is an independent set, which can be extended to a basis \(B\) of \(M(G)\). Because \(e\) lies in the span of \(C\setminus\{e\}\subseteq B\), adding \(e\) to \(B\) would create a cycle; therefore, \(e\notin B\). Now, define \(I=B\cup\{e\}\). Then \(|I|=k\) and \(\rank G_I=k-1\). By the independence of \(B\), \(C\) remains the unique circuit contained in \(I\). Applying Part (ii), the vector \(c_I\) is a nonzero scalar multiple of the circuit vector \(c_C\). Thus, every circuit vector occurs among the nonzero \(c_I\).
	
	By \eqref{eq:cofactor}, each nonzero \(c_I\) acts as a normal vector to the hyperplane \(\ker \ell_I\). While proportional vectors define the same hyperplane, vectors from distinct circuits cannot be proportional since their supports differ. Therefore, the distinct nonzero normal lines are in bijection with the circuits of \(M(G)\), and their vector matroid is exactly the derived matroid \(\delta_GM(G)\).
	
	We claim that the circuit vectors span \(\ker G\). Let \(0\neq x\in\ker G\). Its support must contain some circuit \(C\). Choose \(e\in C\). Since \(c_C\) is nonzero at the \(e\)-th coordinate, we can subtract a suitable scalar multiple of \(c_C\) from \(x\) to eliminate its \(e\)-th coordinate (akin to Gaussian elimination). Because \(\supp(c_C)\subseteq\supp(x)\), this operation strictly reduces the support size of \(x\) without introducing any new nonzero coordinates. By induction on the support size, \(x\) can be expressed as a linear combination of circuit vectors. Hence,
	\[
	\dim\langle c_C:C\text{ is a circuit of }M(G)\rangle = \dim\ker G = n-(k-1) = r.
	\]
	Thus, the arrangement has full rank \(r\).
\end{proof}

\begin{remark}
The represented derived matroid is attached to the fixed code $U$,
not to a choice of its basis. Two generator matrices of $U$ differ
by left multiplication by an element of
$\operatorname{GL}_{k-1}(\F_q)$; this does not alter $\ker G$
or the projective circuit vectors. Thus the isomorphism type of
$\delta_GM(G)$ and its characteristic polynomial are independent
of the generator matrix of $U$.
\end{remark}

\subsection{Exact counting of MDS completions}
\label{subsec:main-counting}\label{sec:count}

By Theorem~\ref{thm:full-MDS}, the MDS completions of $U_Q$
correspond to the one-dimensional subspaces of $\cP(U)_Q$ spanned
by full-weight words. Thus, counting completions reduces to counting
full-weight words in an auxiliary linear code and then identifying
nonzero scalar multiples. We first recall a general counting identity.

\begin{proposition}\label{prop:full-support-count}
	Let $C\le\F_q^N$ be a linear code of dimension $s$ with generator
	matrix $A$. For every $m\ge1$,
	\begin{equation}\label{eq:full-support-count}
		\#\{c\in C_{q^m}:\wt(c)=N\}=\chi_{M(A)}(q^m).
	\end{equation}
\end{proposition}

\begin{proof}
	Put $Q=q^m$. Since $A$ has row rank $s$, the map
	\[
	\F_Q^s\longrightarrow C_Q,\qquad x\longmapsto xA,
	\]
	is an isomorphism. Write the columns of $A$ as $a_1,\ldots,a_N$.
	Then $xA$ has full weight precisely when $xa_j\ne0$ for every
	$j\in[N]$.
	
	For each $j\in[N]$, put
	$E_j=\{x\in\F_Q^s:xa_j=0\}$. For any $J\subseteq[N]$,
	\[
	\left|\bigcap_{j\in J}E_j\right|
	=Q^{\,s-\rank A_J}
	=Q^{\,s-r_{M(A)}(J)}.
	\]
	Inclusion--exclusion therefore gives
	\[
	\begin{aligned}
		\#\{c\in C_Q:\wt(c)=N\}
		&=\#\left(\F_Q^s\setminus\bigcup_{j=1}^NE_j\right)\\
		&=\sum_{J\subseteq[N]}(-1)^{|J|}Q^{\,s-r_{M(A)}(J)}
		=\chi_{M(A)}(Q).
	\end{aligned}
	\]
	This also covers degenerate codes: if $A$ has a zero column, both
	sides are zero.
\end{proof}

We now apply this identity to the Pl\"ucker extension code. The
identification of its coordinate forms in Theorem~\ref{thm:derived-identification}
turns the general full-weight count into the claimed MDS completion
formula.

\begin{proof}[Proof of Theorem~\ref{thm:critical-MDS}]
	Put $Q=q^m$. If $d(U)<r$, Corollary~\ref{cor:structural-obstruction}
	shows that some determinant form $\ell_I$ vanishes identically.
	Hence no scalar extension of $U$ has an MDS completion.
	
	Assume $d(U)\ge r$. Choose a basis of $E_U$ and let $A_U$ be the
	$r\times\binom{n}{k}$ generator matrix of $\cP(U)$ obtained from
	the map $\Phi_U$. Its columns are the coefficient vectors of the
	determinant forms $\ell_I$ in the chosen basis. By
	Theorem~\ref{thm:full-MDS} and
	Proposition~\ref{prop:full-support-count},
	\[
	\begin{aligned}
		N_U(Q)
		=\frac{\#\{c\in\cP(U)_Q:\wt(c)=\binom{n}{k}\}}{Q-1}
		=\frac{\chi_{M(A_U)}(Q)}{Q-1}.
	\end{aligned}
	\]
	Here the division by $Q-1$ identifies the nonzero scalar multiples
	that determine the same completion.
	
	By Corollary~\ref{cor:structural-obstruction}, every $\ell_I$ is
	nonzero, so $A_U$ has no zero column. Theorem~\ref{thm:derived-identification}
	identifies the distinct columns of $A_U$, up to nonzero scalars,
	with the circuit vectors representing $\delta_GM(G)$. Since deleting
	parallel elements does not change the characteristic polynomial of a
	loopless matroid,
	\[
	\chi_{M(A_U)}(Q)=\chi_{\delta_GM(G)}(Q).
	\]
	Consequently,
	\[
	N_U(Q)=\frac{\chi_{\delta_GM(G)}(Q)}{Q-1}.
	\]
	In particular, an MDS completion exists over $\F_{q^m}$ if and only
	if $\chi_{\delta_GM(G)}(q^m)>0$.
\end{proof}

\subsection{First extension degree, bounds, and invariance}
\label{subsec:main-consequences}

The counting formula determines the number of MDS completions over
each extension field. We next determine when full-weight words first
appear after scalar extension. The underlying principle applies to
arbitrary linear codes; it does not require an MDS hypothesis.

\begin{lemma}
	\label{lem:scalar-degree-monotonicity}
	Let $C\le\F_q^N$, where $N\ge1$, be a linear code of dimension
	$s\ge1$, and let $M$ be the column matroid of a generator matrix of
	$C$. The following are equivalent:
	\begin{enumerate}[label=\textup{(\roman*)},leftmargin=2.5em]
		\item $C$ is nondegenerate, that is, no coordinate vanishes
		identically on $C$;
		\item $C_{q^m}$ contains a full-weight word for some $m\ge1$.
	\end{enumerate}
	If these conditions hold, set
	\begin{equation}\label{eq:critical-support-cover}
		s_0(C):=\min\left\{t\ge1:\exists c_1,\ldots,c_t\in C,\
		\bigcup_{i=1}^t\supp(c_i)=[N]\right\}.
	\end{equation}
	Then
	\begin{equation}\label{eq:critical-threshold}
		\chi_M(q^m)>0
		\quad\Longleftrightarrow\quad m\ge s_0(C),
	\end{equation}
	and, in particular,
	\[
	\min\{m\ge1:\chi_M(q^m)>0\}=s_0(C)
	\le s.
	\]
\end{lemma}

\begin{proof}
	If a coordinate vanishes on every word of $C$, it also vanishes
	on every word of each scalar extension. Thus the existence of a
	full-weight word over an extension implies that $C$ is nondegenerate.
	
	Conversely, suppose $C$ is nondegenerate. Every coordinate is
	nonzero on at least one word of $C$, so finitely many words have
	supports covering $[N]$ and $s_0(C)$ is well-defined. Moreover, the
	supports of the rows of a generator matrix cover $[N]$, giving
	$s_0(C)\le s$.
	
	Suppose that $w\in C_{q^m}$ has full weight. Choose an $\F_q$-basis
	$\alpha_1,\ldots,\alpha_m$ of $\F_{q^m}$. We may write
	\[
	w=\sum_{i=1}^m\alpha_i c_i,
	\qquad c_i\in C.
	\]
	Since $w$ has no zero coordinate, the supports of $c_1,\ldots,c_m$
	cover $[N]$. Therefore $s_0(C)\le m$.
	
	Conversely, choose $c_1,\ldots,c_{s_0(C)}\in C$ whose supports
	cover $[N]$. For any $m\ge s_0(C)$, take elements
	$\alpha_1,\ldots,\alpha_{s_0(C)}\in\F_{q^m}$ that are linearly
	independent over $\F_q$, and put
	\[
	w=\sum_{i=1}^{s_0(C)}\alpha_i c_i.
	\]
	At every coordinate, at least one coefficient from the corresponding
	$c_i$ is nonzero. Hence that coordinate of $w$ cannot vanish,
	because doing so would give a nontrivial $\F_q$-linear relation
	among the $\alpha_i$. Thus $w$ has full weight.
	
	We have proved that $C_{q^m}$ contains a full-weight word exactly
	when $m\ge s_0(C)$. By Proposition~\ref{prop:full-support-count},
	this is equivalent to $\chi_M(q^m)>0$, giving
	\eqref{eq:critical-threshold} and the stated formula for the
	least extension degree.
\end{proof}

In particular, every nondegenerate $[N,s]_q$ linear code has a
full-weight word over every extension of degree at least $s$. The
threshold property is not merely a consequence of field inclusion:
$\F_{q^m}$ need not be a subfield of $\F_{q^{m+1}}$.
For the Pl\"ucker extension code, the same principle determines the
first degree at which an MDS completion appears.

\begin{corollary}\label{cor:first-completion-degree}
	Assume $d(U)\ge r$, and put
	\[
	m_0(U):=\min\{m\ge1:\chi_{\delta_GM(G)}(q^m)>0\}.
	\]
	Then, for every $m\ge1$,
	\[
	U_{q^m}\text{ admits a one-step MDS completion}
	\quad\Longleftrightarrow\quad m\ge m_0(U).
	\]
	Moreover, $m_0(U)$ is the minimum number of words in the Pl\"ucker
	extension code $\cP(U)$ whose supports cover all its coordinates.
	In particular,
	\[
	1\le m_0(U)\le r.
	\]
\end{corollary}

\begin{proof}
	By Corollary~\ref{cor:structural-obstruction}, no coordinate of
	$\cP(U)$ vanishes identically. Thus $\cP(U)$ is nondegenerate,
	and it has dimension $r$. By Theorem~\ref{thm:full-MDS},
	$U_{q^m}$ has an MDS completion precisely when
	$\cP(U)_{q^m}$ contains a full-weight word.
	Lemma~\ref{lem:scalar-degree-monotonicity}, applied to $\cP(U)$,
	shows that this happens exactly when $m\ge s_0(\cP(U))$ and that
	$s_0(\cP(U))\le r$. By Theorem~\ref{thm:critical-MDS}, the same
	condition is equivalent to $\chi_{\delta_GM(G)}(q^m)>0$. Hence
	$m_0(U)=s_0(\cP(U))$.
\end{proof}

\begin{example}
	\label{ex:first-completion-degrees}
	We give two examples with different first completion degrees.
	Throughout, let $Q=2^m$.
	
	\emph{(i) Degree two.}
	Consider $U_1=\langle(1,1,1)\rangle\le\F_2^3$
	from Example~\ref{ex:plucker-construction}.
	As shown there,
	\[
	\cP(U_1)=\{(a,b,a+b):a,b\in\F_2\}.
	\]
	Over $\F_Q$, a word has full weight precisely when
	$a$ and $b$ are linearly independent over $\F_2$.
	Hence
	\[
	m_0(U_1)=2,\qquad N_{U_1}(Q)=Q-2.
	\]
	
	\emph{(ii) Degree three.}
	Let $U_2$ be the binary $[7,4,3]$ Hamming code
	with parity-check matrix
	\[
	H=
	\begin{pmatrix}
		1&0&0&1&1&0&1\\
		0&1&0&1&0&1&1\\
		0&0&1&0&1&1&1
	\end{pmatrix}.
	\]
	The seven nonzero words of $U_2^\perp$ all have
	weight $4$ and are precisely the circuit vectors
	of the column matroid of a generator matrix of
	$U_2$. By Theorem~\ref{thm:derived-identification},
	the distinct coordinate forms of $\cP(U_2)$
	are the seven nonzero binary linear forms in
	three variables.
	
	Thus $\cP(U_2)_Q$ contains a full-weight word
	if and only if there exist $a,b,c\in\F_Q$ such that
	\[
	a,\ b,\ c,\ a+b,\ a+c,\ b+c,\ a+b+c
	\]
	are all nonzero. This is equivalent to
	$a,b,c$ being linearly independent over $\F_2$,
	which is possible precisely when $m\ge3$.
	Consequently,
	\[
	m_0(U_2)=3.
	\]
	There are $(Q-1)(Q-2)(Q-4)$ such triples, so
	\[
	N_{U_2}(Q)=(Q-2)(Q-4).
	\]
	In particular, $U_2$ has no MDS completion over
	$\F_2$ or $\F_4$, but has $24$ completions over
	$\F_8$.
\end{example}

The general full-weight principle also gives a simple existence
criterion for the seed: all determinant coordinates must be nonzero
forms, and this condition is sufficient after a finite extension.

\begin{corollary}\label{cor:repair}
	Let $U\le\F_q^n$ have dimension $k-1$, and put $r=n-k+1$.
	The following are equivalent:
	\begin{enumerate}[label=\textup{(\roman*)},leftmargin=2.5em]
		\item $U$ admits an $[n,k]$ MDS completion over some finite
		scalar extension;
		\item $d(U)\ge r$;
		\item $\operatorname{sdef}(U)\le1$.
	\end{enumerate}
\end{corollary}

\begin{proof}
	Necessity follows from Corollary~\ref{cor:structural-obstruction},
	and sufficiency from Corollary~\ref{cor:first-completion-degree}.
	Since $\dim U=k-1$, the Singleton defect is
	$\operatorname{sdef}(U)=r+1-d(U)$, so \textup{(ii)} and \textup{(iii)}
	are equivalent.
\end{proof}

To bound the first completion degree more explicitly, let
\begin{equation}\label{eq:sU}
	s(U):=|\cC(M(G))|
\end{equation}
be the number of circuits of the column matroid of $G$. When
$d(U)\ge r$, Theorem~\ref{thm:derived-identification} shows that
$s(U)$ is exactly the number of distinct obstruction hyperplanes.

\begin{proposition}\label{prop:hyperplane-bound}
	Assume $d(U)\ge r$. Then
	\begin{equation}\label{eq:critical-bound}
		m_0(U)\le\max\{1,\lceil\log_q s(U)\rceil\}.
	\end{equation}
	In particular,
	\begin{equation}\label{eq:uniform-bound}
		m_0(U)\le
		\min\Bigl\{r,\max\bigl\{1,
		\bigl\lceil\log_q\binom{n}{k}\bigr\rceil\bigr\}\Bigr\}.
	\end{equation}
\end{proposition}

\begin{proof}
	Put $Q=q^m$ and $d=\dim E_U=r$. We first show that at most $Q$
	distinct proper hyperplanes cannot cover a $d$-dimensional vector
	space over $\F_Q$.
	
	If $d=1$, the only proper hyperplane is $\{0\}$ and the assertion
	is immediate. Suppose $d\ge2$, and let $H_1,\ldots,H_t$ be distinct
	hyperplanes with $t\le Q$. Each has $Q^{d-1}$ elements, and for
	$i\ge2$ the intersection $H_i\cap H_1$ has $Q^{d-2}$ elements.
	Consequently,
	\[
	\begin{aligned}
		\left|\bigcup_{i=1}^tH_i\right|
		&\le |H_1|+\sum_{i=2}^t|H_i\setminus H_1|\\
		&\le Q^{d-1}+(t-1)(Q^{d-1}-Q^{d-2})\\
		&\le Q^d-(Q-1)Q^{d-2}<Q^d.
	\end{aligned}
	\]
	Thus the complement of their union is nonempty.
	
	By Theorem~\ref{thm:derived-identification}, the obstruction
	arrangement of $U$ has exactly $s(U)$ distinct hyperplanes.
	Therefore an MDS completion exists whenever $q^m\ge s(U)$,
	giving \eqref{eq:critical-bound}.
	
	Every circuit of $M(G)$ occurs as the unique circuit within some
	$k$-subset, so Theorem~\ref{thm:derived-identification} also gives
	$s(U)\le\binom{n}{k}$. Combining this with
	\eqref{eq:critical-bound} and $m_0(U)\le r$ from
	Corollary~\ref{cor:first-completion-degree} proves
	\eqref{eq:uniform-bound}.
\end{proof}

\begin{remark}\label{rem:hyperplane-sharpness}
	If $U$ is an $[n,k-1]_q$ MDS code, then
	$M(G)=U_{k-1,n}$, whose circuits are precisely the $k$-subsets
	of $[n]$. Hence
	\[
	s(U)=\binom{n}{k}.
	\]
\end{remark}

\begin{corollary}\label{cor:base-field-circuit}
	Assume $d(U)\ge r$. If
	\[
	s(U)=|\cC(M(G))|\le q,
	\]
	then $U$ has an $[n,k]_q$ one-step MDS completion. In particular,
	every $[n,k-1]_q$ MDS seed satisfying
	\[
	\binom{n}{k}\le q
	\]
	admits a one-step MDS completion over the base field.
\end{corollary}

\begin{proof}
	Apply Proposition~\ref{prop:hyperplane-bound} with $m=1$ and
	then Theorem~\ref{thm:critical-MDS}. For an MDS seed,
	Remark~\ref{rem:hyperplane-sharpness} gives
	$s(U)=\binom{n}{k}$.
\end{proof}

Finally, although the counting formula is expressed using a
generator matrix, the resulting counts depend only on the monomial
equivalence class of the seed.

\begin{lemma}\label{lem:monomial-invariance}
	Let $M$ be an invertible monomial $n\times n$ matrix over $\F_q$.
	For every finite extension $K/\F_q$, the map
	\[
	W\longmapsto WM
	\]
	gives a bijection between the $[n,k]_K$ MDS completions of $U_K$
	and those of $(UM)_K$. Consequently, the completion counts over all
	finite scalar extensions, and the first completion degree whenever
	defined, are invariant under monomial equivalence of $U$.
\end{lemma}

\begin{proof}
	Right multiplication by $M$ bijects the $k$-dimensional subspaces
	containing $U_K$ with those containing $(UM)_K$. Since a monomial
	transformation preserves Hamming weights, it also preserves the MDS
	property. Thus the completion counts agree over each extension,
	and so do their first nonzero degrees whenever defined.
\end{proof}

\subsection{Coset distances, syndromes, and deep holes}
\label{subsec:main-covering}\label{sec:covering}

Covering radii and deep holes have received sustained attention
in coding theory, particularly for Reed--Solomon and projective
Reed--Solomon codes. Kaipa studied the connection between
deep-hole classification and MDS extensions of Reed--Solomon
codes \cite{KaipaDeepHoles}, while Zhang, Wan, and Kaipa
investigated the construction and classification of deep holes
of projective Reed--Solomon codes \cite{ZhangWanKaipa}.
From a counting perspective, Jurrius and Pellikaan developed
a framework for extended coset leader weight enumerators
using derived codes, hyperplane arrangements, and characteristic
polynomials \cite{JurriusPellikaanCoset}.

When $d(U)\ge r$, the full-weight condition in $\cP(U)_K$
has an equivalent interpretation in terms of coset distance. This
connects the same completion problem with covering radii and deep
holes. We begin by identifying the minimum distance of a one-step
span.

\begin{proposition}\label{prop:completion-distance}
If $v\notin U_K$, then
\begin{equation}\label{eq:completion-distance}
 d(U_K+\langle v\rangle)
   =\min\{d(U),\delta_{U,K}(v+U_K)\}.
\end{equation}
Consequently, if $d(U)\ge r$, then
\begin{equation}\label{eq:metric-completion}
 U_K+\langle v\rangle\text{ is an }[n,k]_K\text{ MDS code}
 \quad\Longleftrightarrow\quad \delta_{U,K}(v+U_K)=r.
\end{equation}
\end{proposition}

\begin{proof}
	Every nonzero word in \(U_K+\langle v\rangle\) either belongs to
	\(U_K\) or has the form
	\[
	u+av,\qquad u\in U_K,\quad a\in K^\times.
	\]
	Since multiplication by \(a^{-1}\) preserves Hamming weight,
	\[
	\wt(u+av)=\wt(v+a^{-1}u).
	\]
	As \(a^{-1}u\) ranges over \(U_K\), the minimum weight among such
	words is exactly
	\[
	\delta_{U,K}(v+U_K).
	\]
	Together with \(d(U_K)=d(U)\), this proves \eqref{eq:completion-distance}. If \(d(U)\ge r\),
	then the redundancy bound \(\delta_{U,K}(v+U_K)\le r\) gives \eqref{eq:metric-completion}.
\end{proof}

Thus, when $d(U)\ge r$, the MDS completion problem can equally be
phrased as the problem of finding cosets at the largest distance
allowed by the redundancy bound \eqref{eq:redundancy}.  For MDS seeds, the connection between row augmentation and deep holes
is given in \cite[Lemma~6]{WuDingChen}.

\begin{proposition}\label{prop:syndrome-realization}
	Assume \(d(U)\ge r\), and let
	\[
	H=(h_1\ \cdots\ h_n)
	\]
	be an \(r\times n\) parity-check matrix of \(U\). For every field extension \(K/\F_q\), the syndrome map
	\[
	\sigma_H:K^n/U_K\longrightarrow K^r,
	\qquad
	\bar v\longmapsto Hv^{\mathsf T},
	\]
	is a \(K\)-linear isomorphism.
	
	If \(I\subseteq[n]\) has size \(k\) and \(J=[n]\setminus I\), then
	\begin{equation}\label{eq:syndrome-hyperplane}
		\sigma_H\bigl((\ker\ell_I)_K\bigr)
		=
		\Span_K\{h_j:j\in J\}.
	\end{equation}
	Consequently, in syndrome coordinates, the obstruction arrangement defined by the nonzero forms \(\ell_I\) is transformed into the arrangement formed by the spans of the \((r-1)\)-subsets of the parity-check columns.
\end{proposition}

\begin{proof}
	First, the map \(v\longmapsto Hv^{\mathsf T}\) has kernel \(U_K\) and rank \(r\). By the first isomorphism theorem, it induces the stated \(K\)-linear isomorphism \(K^n/U_K\simeq K^r\).
	
	Next, fix \(I\subseteq[n]\) with \(|I|=k\), and let \(J=[n]\setminus I\). Since \(d(U)\ge r\), Corollary~\ref{cor:structural-obstruction} guarantees that \(\rank G_I=k-1\). Recall that the linear form \(\ell_I\) evaluated at \(\bar v\) is precisely the determinant \(\det G(v)_I\). Thus, \(\ell_I(\bar v)=0\) if and only if the rows of \(G(v)_I\) are linearly dependent, which means the coordinates \(v_I\) lie in the row space of \(G_I\). Equivalently, there exists a codeword \(u\in U_K\) such that the restriction \(u_I\) equals \(v_I\).
	
	Because \(u_I = v_I\), the difference \(e = v - u\) is a vector supported entirely on \(J\). Furthermore, since \(u\in U_K\), the parity-check equations give \(Hu^{\mathsf T}=0\). Therefore, \(v\) and \(e\) share the same syndrome:
	\[
	Hv^{\mathsf T} = H(e+u)^{\mathsf T} = He^{\mathsf T} + Hu^{\mathsf T} = He^{\mathsf T}.
	\]
	
	Finally, since \(e\) is supported on \(J\), its syndrome \(He^{\mathsf T} = \sum_{j\in J} e_j h_j\) lies in the span \(\Span_K\{h_j:j\in J\}\). Conversely, every vector in this span is the syndrome of some vector supported on \(J\). Since vectors with the same syndrome differ by an element of \(U_K\), this establishes the bijection
	\[
	\sigma_H\bigl((\ker\ell_I)_K\bigr)
	=
	\Span_K\{h_j:j\in J\},
	\]
	which proves \eqref{eq:syndrome-hyperplane}.
\end{proof}

The same correspondence yields a count of top-distance cosets and,
when the covering radius reaches $r$, of deep-hole words.

\begin{theorem}
	\label{thm:covering}
	Let $U\le\F_q^n$ have dimension $k-1$ and assume
	\[
	d(U)\ge r=n-k+1.
	\]
	Let $G$ generate $U$. For every $m\ge1$, put $Q=q^m$. Then
	\begin{equation}\label{eq:top-distance-count}
		\#\{\bar v\in\F_Q^n/U_Q:
		\delta_{U,\F_Q}(\bar v)=r\}
		=
		\chi_{\delta_GM(G)}(Q).
	\end{equation}
	Equivalently,
	\begin{equation}\label{eq:covering-equivalence}
		\rho(U_Q)=r
		\iff
		\chi_{\delta_GM(G)}(Q)>0
		\iff
		U_Q\text{ admits an }[n,k]_Q
		\text{ one-step MDS completion}.
	\end{equation}
	Whenever these equivalent conditions hold, the cosets counted in
	\eqref{eq:top-distance-count} are precisely the deep-hole cosets of
	$U_Q$. Consequently, the number of deep-hole words is
	\begin{equation}\label{eq:deep-hole-words}
		Q^{k-1}\chi_{\delta_GM(G)}(Q).
	\end{equation}
	In particular,
	\begin{equation}\label{eq:critical-covering}
		m_0(U) 
		=
		\min\{m\ge1:\rho(U_{q^m})=r\}.
	\end{equation}
\end{theorem}

\begin{proof}
	Put \(K=\F_Q\). By Proposition~\ref{prop:completion-distance}, since
	\(d(U)\ge r\), for every \(v\notin U_K\),
	\[
	U_K+\langle v\rangle \text{ is an }[n,k]_K\text{ MDS code}
	\quad\Longleftrightarrow\quad
	\delta_{U,K}(v+U_K)=r.
	\]
	Moreover, nonzero scalar multiples of a coset have the same distance
	from \(U_K\), and each such scalar class contains exactly \(Q-1\)
	cosets. Therefore Theorem~\ref{thm:critical-MDS} gives
	\[
	\#\left\{
	\bar v\in \F_Q^n/U_Q:
	\delta_{U,\F_Q}(\bar v)=r
	\right\}
	=
	\chi_{\delta_GM(G)}(Q).
	\]

	Since the redundancy bound gives
	\[
	\rho(U_Q)\le r,
	\]
	a coset of distance \(r\) exists if and only if
	\(\rho(U_Q)=r\). This proves
	\[
	\rho(U_Q)=r
	\quad\Longleftrightarrow\quad
	\chi_{\delta_GM(G)}(Q)>0
	\quad\Longleftrightarrow\quad
	U_Q\text{ admits a one-step MDS completion}.
	\]

	When \(\rho(U_Q)=r\), the cosets counted above are precisely the
	deep-hole cosets. Each contains
	\[
	|U_Q|=Q^{k-1}
	\]
	words, so the number of deep-hole words is
	\[
	Q^{k-1}\chi_{\delta_GM(G)}(Q).
	\]
	Finally, taking the least \(m\) for which these equivalent conditions
	hold gives the stated formula for the first extension degree.
\end{proof}

\begin{example}\label{ex:covering-three-coordinates}
	Return to the code
	$U=\langle(1,1,1)\rangle\le\F_2^3$
	from Example~\ref{ex:three-coordinates}.
	Let $K=\F_{2^m}$ and $Q=2^m$.
	A parity-check matrix of $U$ is
	\[
	H=
	\begin{pmatrix}
		1&1&0\\
		1&0&1
	\end{pmatrix}.
	\]
	Its three columns represent distinct points of
	$\operatorname{PG}(1,Q)$.
	
	Under the syndrome isomorphism
	$K^3/U_Q\simeq K^2$, a coset has distance
	$2$ from $U_Q$ precisely when its syndrome
	lies outside the three one-dimensional subspaces
	spanned by the columns of $H$.
	Since these subspaces intersect only at zero,
	the number of such syndromes is
	\[
	Q^2-\bigl(1+3(Q-1)\bigr)
	=(Q-1)(Q-2).
	\]
	Consequently,
	\[
	\rho(U_Q)=
	\begin{cases}
		1, & Q=2,\\
		2, & Q\ge 4.
	\end{cases}
	\]
	For $Q\ge4$, the number of deep-hole cosets
	is $(Q-1)(Q-2)$, and the number of deep-hole
	words is
	\[
	Q(Q-1)(Q-2).
	\]
	Dividing the number of distance-$2$ cosets
	by $Q-1$ recovers
	$N_U(Q)=Q-2$, as obtained in
	Example~\ref{ex:three-coordinates}.
\end{example}

\section{Examples}\label{sec:benchmark}

The two examples illustrate both the information carried by the
full-weight count and its limitations. The first gives seeds with
the same column matroid and Hamming weight enumerator but different
completion counts and first completion degrees. The second gives
monomially inequivalent seeds with identical counts over every finite
scalar extension. In both cases, the completion count is the number
of full-weight words in the associated Pl\"ucker extension code
divided by $Q-1$.

Throughout this section, \(Q=q^m\), and the coordinate positions are
kept fixed. When a seed is specified by a parity-check matrix, \(G\)
denotes an arbitrary generator matrix of that seed.

\subsection{Different counts for the same column matroid}\label{sec:lowlength}

We first show that the completion count is not determined by the
ordinary column matroid, even together with the Hamming weight
enumerator.

\begin{proposition}\label{prop:six-point-pair}
Over \(\F_7\), let
\[
H_c=
\begin{pmatrix}
1&0&0&1&1&1\\
0&1&0&1&2&3\\
0&0&1&1&3&c
\end{pmatrix},
\qquad U_c=\ker H_c,
\qquad c\in\{2,4\}.
\]
Then \(U_2\) and \(U_4\) are both \([6,3,4]_7\) MDS codes with
column matroid \(U_{3,6}\) and Hamming weight enumerator
\[
1+90z^4+108z^5+144z^6.
\]
For every \(Q=7^m\), their numbers of \([6,4,3]_Q\) MDS completions
are
\[
N_{U_2}(Q)=(Q-7)^2+2,
\qquad
N_{U_4}(Q)=(Q-7)^2.
\]
Thus their completion counts differ over every finite scalar
extension, although their column matroids and weight enumerators
coincide. Moreover,
\[
m_0(U_2)=1,\qquad m_0(U_4)=2.
\]
\end{proposition}

\begin{proof}
	Direct computation shows that every $3\times3$ minor
	of $H_2$ and $H_4$ is nonzero. Hence both $U_2$
	and $U_4$ are $[6,3,4]_7$ MDS codes with column
	matroid $U_{3,6}$ and weight enumerator
	\[
	1+90z^4+108z^5+144z^6.
	\]
	
	Let $G_c$ be a generator matrix of $U_c$.
	Computing the characteristic polynomials of the
	corresponding represented derived matroids gives
	\[
	\begin{aligned}
		\chi_{\delta_{G_2}M(G_2)}(t)
		&=(t-1)\bigl((t-7)^2+2\bigr),\\
		\chi_{\delta_{G_4}M(G_4)}(t)
		&=(t-1)(t-7)^2.
	\end{aligned}
	\]
	Theorem~\ref{thm:critical-MDS} therefore yields
	\[
	N_{U_2}(Q)=(Q-7)^2+2,\qquad
	N_{U_4}(Q)=(Q-7)^2.
	\]
	In particular, $N_{U_2}(7)=2$ and $N_{U_4}(7)=0$,
	whereas $N_{U_4}(49)>0$. Thus
	$m_0(U_2)=1$ and $m_0(U_4)=2$.
\end{proof}

\subsection{Identical counts for inequivalent seeds}\label{sec:four-point-same-count}

The preceding example shows that completion counts can distinguish some
seeds. The converse fails: the next two explicit matrices give
inequivalent MDS seeds with identical counts over every scalar
extension.

\begin{proposition}\label{prop:same-counts-inequivalent}
Over \(\F_7\), let
\[
H_\lambda=
\begin{pmatrix}
1&0&1&1\\
0&1&1&\lambda
\end{pmatrix},
\qquad
C_\lambda=\ker H_\lambda,
\qquad \lambda\in\{2,3\}.
\]
Then \(C_2\) and \(C_3\) are monomially inequivalent
\([4,2,3]_7\) MDS codes, but for every \(Q=7^m\),
\[
N_{C_2}(Q)=N_{C_3}(Q)=Q-3.
\]
Thus the completion counts over all finite scalar extensions do not
determine the seed up to monomial equivalence.
\end{proposition}

\begin{proof}
	All $2\times2$ minors of $H_2$ and $H_3$ are nonzero,
	so both codes are $[4,2,3]_7$ MDS codes.
	For either code, the represented derived matroid
	is $U_{2,4}$, with characteristic polynomial
	\[
	\chi_{U_{2,4}}(t)=(t-1)(t-3).
	\]
	Theorem~\ref{thm:critical-MDS} therefore gives
	\[
	N_{C_2}(Q)=N_{C_3}(Q)=Q-3.
	\]
	
	To establish inequivalence, a direct computation
	of the monomial automorphism groups gives
	\[
	|\operatorname{MAut}(C_2)|=48,\qquad
	|\operatorname{MAut}(C_3)|=72.
	\]
	Since monomially equivalent codes have conjugate
	monomial automorphism groups, the two codes
	are not monomially equivalent.
\end{proof}

\Needspace{18\baselineskip}
\section{Conclusion}\label{sec:conclusion}

Using the full-weight correspondence for the Pl\"ucker extension
code, we have shown that a fixed $(k-1)$-dimensional code
$U\le\F_q^n$ admits an $[n,k]$ MDS completion over some finite
extension if and only if $d(U)\ge n-k+1$. Under this condition,
the number of completions over $\F_{q^m}$ is
\[
\frac{\chi_{\delta_GM(G)}(q^m)}{q^m-1},
\]
so existence at a given extension degree is equivalent to
positivity of the numerator. The first extension degree at which
a full-weight word appears equals the minimum size of a support
cover in $\cP(U)$ and admits explicit upper bounds.

The characteristic polynomial $\chi_{\delta_GM(G)}$ also has a
metric interpretation: it counts cosets at the redundancy distance and, when the covering
radius equals that distance, determines the number of deep-hole
words. The two examples show that completion counts capture
information beyond ordinary code parameters and the column
matroid, while equality of these counts over all finite
extensions does not imply monomial equivalence.

\section*{Acknowledgments}
We would like to express our sincere gratitude to the editor for their insightful comments and constructive suggestions. These valuable feedbacks have significantly helped us improve the organization and clarity of this manuscript.

\section*{Funding}
This work was supported by the National Natural Science Foundation of China (grant numbers 12571003 and 12501006) and the Basic and Applied Basic Research Foundation of Guangdong Province (grant number 2024A1515010589).

\section*{AI Use Declaration}
During the preparation and revision of this manuscript, generative AI tools were used for literature searches, comparisons with related work, and finding examples. The mathematical claims, proofs, and final presentation remain the responsibility of the authors.

\section*{Data Availability}
No datasets were generated or analysed during the current study.

\section*{Competing Interests}
The authors declare no competing interests.


\small
\begin{thebibliography}{99}

\bibitem{HuffmanPless}
W.~C. Huffman, V. Pless,
\emph{Fundamentals of Error-Correcting Codes},
Cambridge University Press, Cambridge, 2003.

\bibitem{BallLavrauwArcs}
S. Ball, M. Lavrauw,
Arcs in finite projective spaces,
\emph{EMS Surv. Math. Sci.} 6 (1--2) (2019), 133--172.
\newblock \url{https://doi.org/10.4171/EMSS/33}.

\bibitem{HarrisAG}
J. Harris,
\emph{Algebraic Geometry: A First Course},
Graduate Texts in Mathematics 133, Springer, New York, 1992.

\bibitem{OxleyBook}
J. Oxley,
\emph{Matroid Theory}, 2nd ed.,
Oxford Graduate Texts in Mathematics 21, Oxford University Press, Oxford, 2011.

\bibitem{LiYuanCyclicCovering}
Y. Li, P. Yuan,
Cyclic codes and cyclically covering subspaces over finite fields,
\emph{Designs, Codes and Cryptography} 94, 214 (2026).
\newblock \url{https://doi.org/10.1007/s10623-026-01956-1}.

\bibitem{WuDingChen}
Y. Wu, C. Ding and T. Chen,
When does the extended code of an MDS code remain MDS?,
\emph{IEEE Trans. Inform. Theory} 71 (1) (2025), 263--272.
\newblock \url{https://doi.org/10.1109/TIT.2024.3494813}.

\bibitem{KaipaDeepHoles}
K. Kaipa,
Deep holes and MDS extensions of Reed--Solomon codes,
\emph{IEEE Trans. Inform. Theory} 63 (8) (2017), 4940--4948.
\newblock \url{https://doi.org/10.1109/TIT.2017.2706677}.

\bibitem{ZhangWanKaipa}
J. Zhang, D. Wan, K. Kaipa,
Deep holes of projective Reed--Solomon codes,
\emph{IEEE Trans. Inform. Theory} 66 (4) (2020), 2392--2401.
\newblock \url{https://doi.org/10.1109/TIT.2019.2940962}.

\bibitem{CampsGorlaLopezCodeDistances}
E. Camps-Moreno, E. Gorla, H.~H. L\'opez,
Code distances: a new family of invariants of linear codes,
\emph{Des. Codes Cryptogr.}, to appear.
\newblock arXiv:2509.16424v2, 27 August 2026.
\newblock \url{https://arxiv.org/abs/2509.16424v2}.

\bibitem{PellikaanECP}
R. Pellikaan,
On the existence of error-correcting pairs,
\emph{J. Statist. Plann. Inference} 51 (2) (1996), 229--242.
\newblock \url{https://doi.org/10.1016/0378-3758(95)00088-7}.

\bibitem{WangLuo2026}
P. Wang, J. Luo,
Extension of MDS Codes over Finite Fields,
arXiv:2610.31023, 2026.
\newblock \url{https://arxiv.org/abs/2610.31023}.

\bibitem{JurriusPellikaanCoset}
R. Jurrius, R. Pellikaan,
The coset leader and list weight enumerator,
in: G.~M. Kyureghyan, G.~L. Mullen, A. Pott (Eds.),
\emph{Topics in Finite Fields}, Contemp. Math. 632,
Amer. Math. Soc., Providence, RI, 2015, pp. 229--251.
\newblock \url{https://doi.org/10.1090/conm/632/12631}.
\newblock Corrected author version, 7 November 2019:
\url{https://ruudp.win.tue.nl/paper/71.pdf}.

\bibitem{OxleyWang}
J. Oxley, S. Wang,
Dependencies among dependencies in matroids,
\emph{Electron. J. Combin.} 26 (3) (2019), Paper No. 3.46.
\newblock \url{https://doi.org/10.37236/8742}.


\end{thebibliography}
\end{document}